\documentclass[journal]{IEEEtran}

\ifCLASSOPTIONcompsoc
  \usepackage[caption=false,font=normalsize,labelfont=sf,textfont=sf]{subfig}
\else
  \usepackage[caption=false,font=footnotesize]{subfig}
\fi

\usepackage{float}
\usepackage{color}
\usepackage{cite}
\usepackage{amssymb,amsmath}
\usepackage{graphicx}
\usepackage[lined,ruled,commentsnumbered,linesnumbered]{algorithm2e}
\usepackage{color}
\usepackage{url}
\usepackage{placeins}

\usepackage{etoolbox}

\usepackage{comment}

\usepackage{blindtext}

\usepackage{multirow}

\usepackage{scrextend}
\addtokomafont{labelinglabel}{\sffamily}

\usepackage[official]{eurosym}

\usepackage{amsthm}
\newtheorem{definition}{Definition}
\newtheorem{theorem}{Theorem}

\begin{document}


\title{A Secure and Privacy-preserving Protocol for\\ Smart Metering Operational Data Collection}

\author{Mustafa~A.~Mustafa,~Sara~Cleemput,~Abdelrahaman~Aly,~and~Aysajan~Abidin

\thanks{M.A. Mustafa is with the School of Computer Science, The University of Manchester, UK. e-mail: (mustafa.mustafa@manchester.ac.uk).}
\thanks{S. Cleemput, A. Aly and A. Abidin are with the imec-COSIC research group, KU Leuven, Belgium. e-mail: (\{sara.cleemput, aysajan.abidin, abdelrahaman.aly\}@esat.kuleuven.be).}
}

\maketitle


\begin{abstract}
Smart grid allows fine-grained smart metering data collection which can improve the efficiency and reliability of the grid. Unfortunately, this vast collection of data also impose risks to users' privacy. In this paper, we propose a novel protocol that allows suppliers and grid operators to collect users' aggregate metering data in a secure and privacy-preserving manner. We use secure multiparty computation to ensure privacy protection. In addition, we propose three different data aggregation algorithms that offer different balances between privacy-protection and performance. Our protocol is designed for a realistic scenario in which the data need to be sent to different parties, such as grid operators and suppliers. Furthermore, it facilitates an accurate calculation of transmission, distribution and grid balancing fees in a privacy-preserving manner. We also present a security analysis and a performance evaluation of our protocol based on existing multiparty computation algorithms. 

\end{abstract}


\IEEEpeerreviewmaketitle


\section{Introduction}
\label{Introduction}

\IEEEPARstart{S}{mart} Grid (SG) is the electrical grid of the future, adding a communication network to the traditional electrical grid  infrastructure. This allows bidirectional communication between the different entities and components of the grid, facilitating automated grid management. The overall aim is to make the electrical grid more reliable and efficient~\cite{Farhangi2010}. This is achieved by automatically collecting fine-grained metering data from Smart Meters (SMs), which replace the traditional electricity meters. These metering data include electricity consumption and production measurements. Electricity production takes place if households own a Distributed Energy Resource (DER), e.g., solar panels. All these data are sent to the grid operators and suppliers at regular intervals, e.g., 15 minutes.


Access to fine-grained metering data gives entities two main advantages. Firstly, these data allow suppliers to predict their customers' electricity consumption and production more accurately. These patterns are essential to allow the supplier to predict the amount of electricity it needs to buy on the wholesale market for every trading period. Since suppliers pay heavy imbalance fees for every deviation of the actual consumption compared to their prediction, it is crucial for them to obtain accurate consumption and production patterns. Secondly, fine-grained metering data also allow accurate settling of all the fees after each trading period, which is essential to realise local electricity trading markets~\cite{Abidin2016, Abidin2018secure}. Currently, the imbalance fees for the suppliers are calculated proportional to their number of customers in each neighbourhood - i.e., households connected to the same feeder. The current imbalance fee is an estimate. With SM data, accurate settling of fees becomes possible. The same is true for the distribution and transmission fees which suppliers pay to the Distribution Network Operator (DNO) and Transmission System Operator (TSO).


Unfortunately, fine-grained metering data also have disadvantages: they pose a privacy threat to users. Any entity who has access to individual users' fine-grained metering data can use non-intrusive load monitoring techniques~\cite{hart} to analyse consumption patterns and infer user activities~\cite{Kalogridis:USaPP}. As a simple example of how such a  privacy invasion can lead to adverse consequences, consider an insurance company which increases the insurance fee if they observe from the consumption pattern that their customers do not have a healthy life style. The Netherlands have even abandoned their planned mandatory roll-out of SMs because of the privacy issues\cite{Heck}.

The UK has privacy protection as a requirement for their smart metering architecture~\cite{UK_DCC}. However, their proposed architecture contains a centralised entity, the Data Communications Company (DCC), that collects all metering data and provides a privacy-friendly version of it to authorised entities. Although this might ensure privacy protection against these entities (if the privacy-friendly version is properly generated), it does not protect against the DCC which has access to all users' data. 

There are two main approaches for user privacy protection: anonymisation and data aggregation. Proper anonymisation is difficult to achieve, as de-anonymisation is almost always possible~\cite{Cleemput2018}. Aggregation is a better approach, but the current proposals~\cite{Kursawe2011,Li2010,Garcia2011} still have shortcomings: (i) they are designed for system models in which data are sent to only one entity, thus they are not applicable to current electricity markets, (ii) they do not consider electricity generated by residential DERs and injected to the grid, and (iii) they do not support transmission, distribution and balancing fee calculation. The real challenge is to design solutions that protect users' privacy not only from external entities but also from internal ones, and that are efficient, fault-tolerant and applicable to real-world smart metering architecture.

In this paper we propose a secure and privacy-preserving protocol for collecting metering data. This work extends our previous research~\cite{Mustafa2017} by improving the data aggregation algorithm. Our main contribution is twofold:


\begin{itemize}

\item We design a secure and privacy-preserving protocol for collecting operational metering data which is required for calculating distribution, transmission and imbalance fees. Our protocol uses Multiparty Computations (MPC) as the underlying cryptographic primitive and supports three different privacy-friendly data aggregation algorithms. Additionally, it supports realistic system models (with multiple data recipients of aggregates of various subsets of users' metering data); it is fault-tolerant; it is applicable to existing liberalised market models, and it also supports electricity production data generated by users. 

\item We analyse the computational complexity and communication cost of our protocol in a realistic setting based on the UK's smart metering architectue~\cite{UK_DCC}. 

\end{itemize}

The remainder of the paper is organised as follows: Section~\ref{Related Work} discusses the related work, Section~\ref{Preliminaries} gives the necessary preliminaries, Section~\ref{Privacy-preserving Smart Metering Protocols} proposes a protocol (and three aggregation algorithms) for secure and privacy-preserving operational metering data collection. Sections~\ref{Security Analysis}~and~\ref{Performance Evaluation} analyse its security and privacy properties, and evaluate its performance, respectively. Finally, Section~\ref{Conclusions} concludes the paper.


\section{Related Work}
\label{Related Work}

Security and privacy concerns in SG have been raised~\cite{Kalogridis:USaPP} and various protocols have already been proposed~\cite{Efthymiou2010, Li2010, Kursawe2011, MUSP, DEP2SA, Danezis2013, Rottondi2013IEEE, Garcia2011, Mustafa2017}. To protect users' privacy, these protocols usually take two approaches: anonymisation or aggregation. Efthymiou and Kalogridis~\cite{Efthymiou2010} proposed that each SM also has an anonymous ID for reporting only operational metering data. However, Cleemput~et~al.~\cite{Cleemput2018} showed that de-anonymisation is possible. 

To achieve privacy-friendly aggregation, Li~et~al.~\cite{Li2010} proposed to use homomorphic encryption. However, their protocol does not protect against active attackers nor facilitate current electricity markets. Mustafa~et~al.~\cite{MUSP, DEP2SA} addressed these limitations by using digital signatures and a selective data aggregation and delivery method. Garcia and Jacobs~\cite{Garcia2011} combined homomorphic encryption with data sharing to allow the data recipient to aggregate the data. The use of homomorphic encryption can protect users' privacy, but it also introduces high computational costs to SMs. To overcome this limitation, Kursawe~et~al.~\cite{Kursawe2011} proposed a lightweight aggregation scheme which requires SMs to mask their data with noise that cancels out when added together. The scheme is computationally efficient, but it requires a complex reinitialization process when adding SMs and does not support flexible aggregation groups. Rahman~et~al.~\cite{Rahman2017} proposed a distributed aggregation of metering data where an initialisation SM adds random value to its metering data before sending the sum to the next SM in the group. In turns, each SM adds their data to the accumulated sum before the value is returned to the initial SM that extracts the beforehand added random value to obtain the aggregate data of all the SMs. However, the data of a SM could easily be learnt if two of the neighbouring SMs in the ring collaborate.

Gope~and~Sikdar~\cite{Gope2018} proposed a privacy-preserving aggregation scheme for billing and demand response management, named EDAS, that uses random values to mask the individual measurements of SMs. However, their scheme only partially protects users' privacy, i.e., the service provider is fully trusted, thus it learns all the users metering data, and the aggregator also learns the aggregate data of users located in the same region. Liu~et~al.~\cite{Liu2018} proposed a data aggregation scheme, named 3PDA, that uses a virtual aggregation area to mask the metering data of individual users. Although 3PDA does not rely on a trusted party to aggregate the data, its system model is simplified (it has only a service provider, aggregator and SMs), thus not practical to deploy in existing grid architecture. Lyu~et~al.~\cite{Lyu2018} proposed a fog-enabled scheme, named PPFA. PPFA uses nodes that are closer to SMs to perform an initial aggregation of the data, and then the service provider uses the aggregate data provided by the fog nodes to extract statistics about the users' metering data. Although the obtained results protect users' privacy, i.e., they are differentially private, such results can not be used for billing and settlements as they are just (close) estimations of the real consumption. 

He~et~al.~\cite{He2016} proposed a privacy-preserving aggregation scheme that also protects against internal attackers. However, their scheme uses bilinear mapping, which is computationally expensive for devices with limited resources. He~et~al.~\cite{he2017efficient, He2017} proposed modified schemes that outperform~\cite{He2016} in terms of computation and communication costs. Knirsch~et~al.~\cite{Knirsch2017} combined the schemes proposed in~\cite{Erkin2015}~and~\cite{Rane2015} to construct a scheme for privacy-friendly fault-tolerant aggregation over multiple sets of SMs. However, the scheme is not efficient, as it uses computationally expensive homomorphic techniques as well as a full round between SMs and the service provider. Li~et~al.~\cite{Li2018} proposed a privacy-preserving multi-subset data aggregation scheme, named PPMA. Their scheme combines two super-increasing sequences (similar to EPPA~\cite{Lu2012}) and Paillier cryptosystem. However, the SMs of each subset are restricted to the ones whose consumption data is within a predefined range, limiting the usefulness of the aggregate data. 

Knirsch~et~al.~\cite{Knirsch2015} use masking to achieve error-resilient data aggregation. Although their scheme supports aggregation over multiple sets of SMs, it adds additional communication costs as SMs communicate between each other to exchange their masking values. Abdallah~and~Shen~\cite{Abdallah2017, Abdallah2018} proposed a lattice-based homomorphic data aggregation scheme, whereas Shen~et~al.~\cite{Shen2017} proposed a cube-data aggregation scheme that allows aggregation of multi-dimensional data to obtain sums of each dimension without revealing users' private information. However, these schemes do not support aggregation of multiple sets of SMs. Borges~and~M\"{u}hlh\"{a}user~\cite{EPPP4SMS} proposed a homomorphic encryption-based protocol, where each SM has two secret keys, and an authorised data recipient is aware of the sums of all the corresponding keys of all the SMs in the network. Each SM encrypts its metering data using its keys and sends the ciphertext directly to the data recipient. The data recipient  aggregates all the ciphertexts and then using the sums of the keys, it can decrypt only the sum of the SMs' metering data. The protocol has lower computational costs at SMs (the costs of encryption), compared to the  original  Paillier  scheme. However, the decryption cost at the data recipient is high. Also, the communication overheads are high, as each SM sends its data to the data recipient, i.e., the aggregation is done at the data recipient (the destination of the data), not in the network. Tonyali~et~al.~\cite{TONYALI2018} proposed mechanisms to reduce communication overheads when somewhat homomorphic encryption is used to perform in-network data aggregation.

Another approach to aggregate data in a privacy-preserving (and efficient) manner is MPC. Danezis~et~al.~\cite{Danezis2013} proposed protocols using secret-sharing based MPC to detect fraud and to extract advanced grid statistics. Rottondi~et~al.~\cite{Rottondi2013IEEE} proposed a novel security architecture for aggregation of metering data. However, their architecture requires additional nodes in the system, i.e., gateways placed at the users' households. 


    
    
    

Unlike the aforementioned work, our proposed MPC-based privacy-preserving protocol for operational metering data collection (i) is based on a real smart metering architecture, (ii) is readily applicable to a liberalised electricity market with various stakeholders, (iii) takes into account not only the electricity consumption data, but also electricity injected into the grid by households, and (iv) allows the TSO, DNOs and suppliers to calculate the exact distribution, transmission and balancing fees based on real data rather than on estimates. Furthermore, our protocol is fault-tolerant and it protects users' privacy against (colluding) internal adversaries. 


\section{Preliminaries}
\label{Preliminaries}


\subsection{System Model}


\begin{figure}[!t]
\centering
\includegraphics[trim= 0 0 0 0,clip=true,width=3.43in]{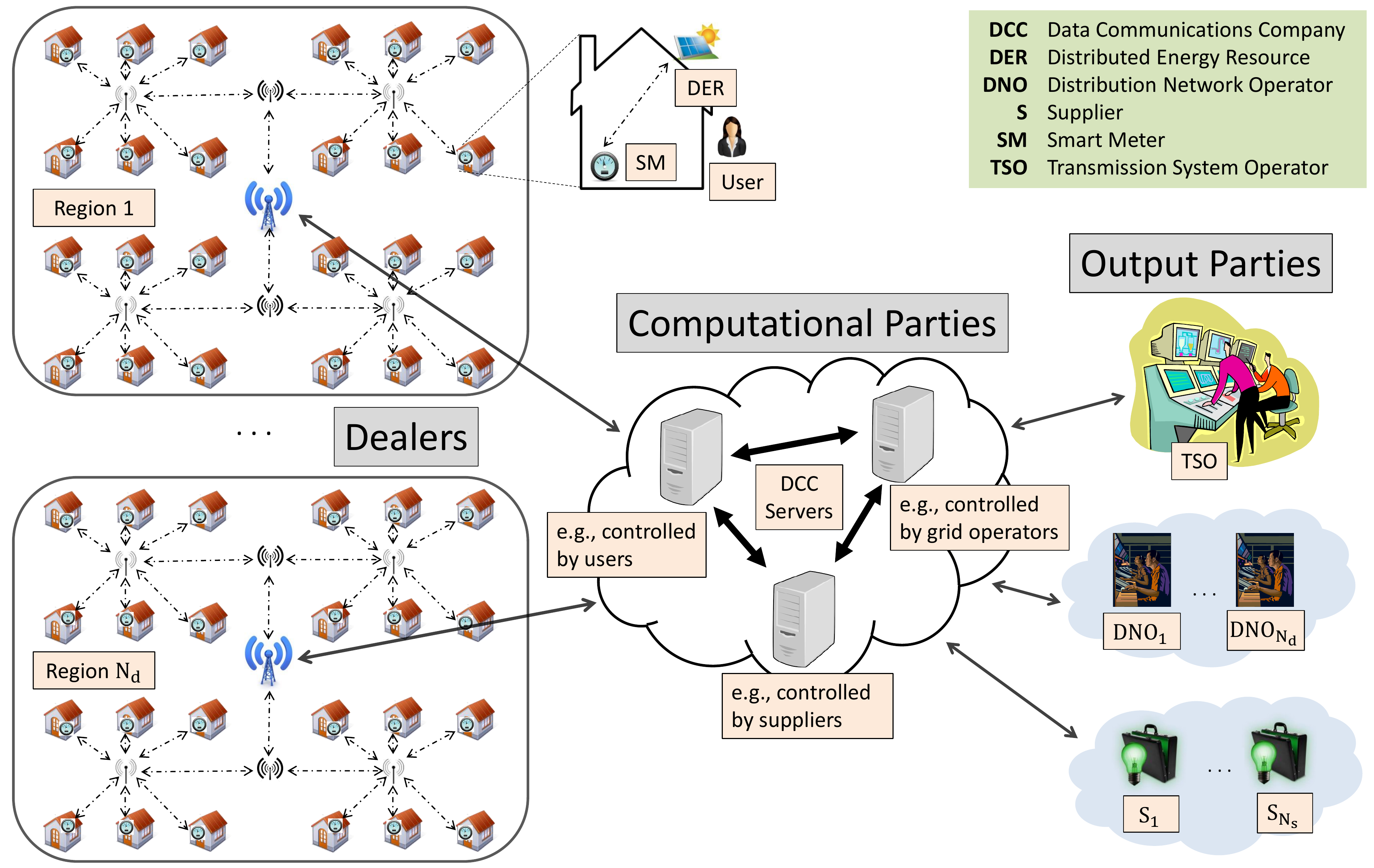}
\caption{System model.}
\label{fig: System Model in AMI}
\end{figure}


As shown in Fig.~\ref{fig: System Model in AMI}, our system model follows the smart metering architecture used in the UK~\cite{UK_smart_metering}, and consists of the following entities. \textit{Users} consume electricity and are billed for this by their contracted supplier. \textit{Distributed Energy Resources} (\textit{DERs}) are mini generators (e.g., solar panels) located on users' premises. Most of the electricity they generate is consumed by their owners. However, surplus electricity may be injected into the grid. \textit{Smart Meters} (\textit{SMs}) are advanced metering devices that measure the electricity flowing from the grid to the house and vice versa per time slot, $t_k$. SMs regularly communicate with other authorised SG entities. \textit{Suppliers} are responsible for supplying electricity to all users. 
They buy this electricity from generators on the wholesale market, and sell it to users. They are also obliged to buy any electricity their customers inject into the grid. If a supplier buys an incorrect amount of electricity, it will be punished with imbalance fees. \textit{Distribution Network Operators} (\textit{DNOs}) are responsible for managing and maintaining the electricity distribution lines (i.e., the low/middle voltage lines) in their respective regions. They charge suppliers distribution fees based on the electricity consumption of the suppliers' customers in each time slot. The suppliers then charge their customers this fee in turn. \textit{Transmission System Operator} (\textit{TSO}) is responsible for managing and maintaining the electricity transmission lines (i.e., the high voltage lines) in the grid as well as balancing the whole grid at any point in time. It charges suppliers transmission and balancing fees based on the consumption of their customers in each time slot. Similarly, suppliers pass this cost to users. \textit{Data Communications Company} (\textit{DCC}) is a centralised entity that consists of servers run by different parties. It collects and delivers users' metering data to the TSO, DNOs and suppliers. 

We also classify some of these entities into three groups: 
\textit{dealers} (i.e., SMs) who provide the input data, \textit{computational parties} (i.e., DCC servers) who perform computations on the input data, and \textit{output parties} (i.e., TSO, DNOs, and suppliers) who receive the results of the computations. 
The SMs generate and provide the DCC servers with input data including the electricity consumption and generation data measured per time slot. The DCC servers  obtain input data from the SMs, jointly perform the necessary calculations and provide the TSO, DNOs and suppliers with the results of these calculations.


\subsection{Threat Model and Assumptions} 
\label{Threat Model and Assumptions}

All external and internal entities (i.e., users, DCC servers, TSO, DNOs and suppliers) can act maliciously. External entities may eavesdrop or modify data in transit trying to gain access to confidential data or to disrupt the SG. Any user may try to modify metering data sent by their SMs in an attempt to gain financial advantage or learn other users' data. The TSO, DNOs and suppliers may manipulate users' metering data in an attempt to gain financial advantage, i.e., to manipulate the transmission/distribution fees and imbalance fines calculations. They may also try to learn individual users' consumption data or the aggregate consumption of any group of users located in different regions or contracted by competitors. Any DCC server may manipulate the metering data it receives/computes.

As our main concern is privacy protection of users, we make the following realistic assumptions. The communication channels among all entities are encrypted and authenticated. At least one DCC server is honest and trustworthy, i.e., it follows the protocol specifications. SMs are tamper-proof and sealed, thus no one can tamper with them without being detected. This is a common assumption in the community. If an SM is faulty or compromised (sending false data), there are standard techniques (e.g., use control meters) to trace back and identify this meter. However, this is out of the scope of this paper.

\subsection{Notations}
\label{Notations}

We denote the SM of household $i$ as ${\textnormal{SM}_i} \in \mathbb{SM}$, where $\mathbb{SM}$ is the set of all the SMs in the grid of a country, and the amount of electricity taken from the grid (i.e., imported electricity) and the amount fed back to the grid (i.e., exported electricity) by household $i$ during the $k$th time slot, $t_k$, as $\textnormal{E}_i^{\textnormal{imp},t_k} \in~\mathbb{E}^{\textnormal{imp},t_k}$ and $\textnormal{E}_i^{\textnormal{exp},t_k} \in~\mathbb{E}^{\textnormal{exp},t_k}$, respectively. $\mathbb{E}^{\textnormal{exp},t_k}$ and $\mathbb{E}^{\textnormal{exp},t_k}$ are the aggregate of consumption data and electricity fed back data, respectively, of all the households during $t_k$ in the grid. More notations are given in Table~\ref{table Notations}. Also, $\sum(X)$ denotes the aggregate (sum) value of all elements in set $X$. 


\begin{table}[!t]
\caption{Notations}
\label{table Notations}
\centering
\begin{tabular}{l p{6.60cm}}
\hline
Symbol & Meaning \\
\hline

$t_k$ & $k$th time slot, $k=1, \ldots, {\textnormal{N}_{\textnormal{t}}}$ \\ 

$\textnormal{d}_{j}$ & the DNO operating in region $j$, $j=1, \ldots, {\textnormal{N}_{\textnormal{d}}}$\\

$\textnormal{s}_{u}$ & $u$th supplier, $u=1, \ldots, {\textnormal{N}_{\textnormal{s}}}$ \\

$\textnormal{SM}_{i}$ & the SM belonging to household $i$ \\ 

$\mathbb{SM}$ & set of all the SMs in a specific country \\

$\mathbb{SM}_{\textnormal{d}_j}$ & set of all the SMs operated by DNO $\textnormal{d}_j$ \\

$\mathbb{SM}^{\textnormal{imp}}_{\textnormal{s}_u}$ & set of all the SMs whose users buy electricity from $\textnormal{s}_u$ \\

$\mathbb{SM}^{\textnormal{exp}}_{\textnormal{s}_u}$ & set of all the SMs whose users sell electricity to $\textnormal{s}_u$ \\

$\mathbb{SM}^{\textnormal{imp}}_{\textnormal{d}_j, \textnormal{s}_u}$ & set of all the SMs operated by $\textnormal{d}_j$ and whose users buy electricity from $\textnormal{s}_u$ \\

$\mathbb{SM}^{\textnormal{exp}}_{\textnormal{d}_j, \textnormal{s}_u}$ & set of all the SMs operated by $\textnormal{d}_j$ whose users sell electricity to $\textnormal{s}_u$ \\



$\textnormal{E}_i^{\textnormal{imp},{t_k}}$ & amount of electricity imported by household $i$ during $t_k$ \\

$\textnormal{E}_i^{\textnormal{exp},t_k}$ & amount of electricity exported by household $i$ during $t_k$ \\

$\mathbb{E}^{\textnormal{imp},t_k}$ & aggregate data of all $\textnormal{E}_i^{\textnormal{imp},t_k}$ for $\textnormal{SM}_i \in \mathbb{SM}$ \\

$\mathbb{E}^{\textnormal{exp},t_k}$ & aggregate data of all $\textnormal{E}_i^{\textnormal{exp},t_k}$ for $\textnormal{SM}_i \in \mathbb{SM}$ \\

$\mathbb{E}_{\textnormal{d}_j}^{\textnormal{imp},t_k}$ & aggregate data of all $\textnormal{E}_i^{\textnormal{imp},t_k}$ for $\textnormal{SM}_i \in  \mathbb{SM}_{\textnormal{d}_j}$ \\

$\mathbb{E}_{\textnormal{d}_j}^{\textnormal{exp},t_k}$ & aggregate data of all $\textnormal{E}_i^{\textnormal{exp},t_k}$ for $\textnormal{SM}_i \in  \mathbb{SM}_{\textnormal{d}_j}$ \\

$\mathbb{E}_{\textnormal{s}_u}^{\textnormal{imp},t_k}$ & aggregate data of all $\textnormal{E}_i^{\textnormal{imp},t_k}$ for $\textnormal{SM}_i \in~ \mathbb{SM}^{\textnormal{imp}}_{\textnormal{s}_u}$ \\

$\mathbb{E}_{\textnormal{s}_u}^{\textnormal{exp},t_k}$ & aggregate data of all $\textnormal{E}_i^{\textnormal{exp},t_k}$ for $\textnormal{SM}_i \in~\mathbb{SM}^{\textnormal{exp}}_{\textnormal{s}_u}$ \\

$\mathbb{E}_{\textnormal{d}_j,\textnormal{s}_u}^{\textnormal{imp},t_k}$ & aggregate data of all $\textnormal{E}_i^{\textnormal{imp},t_k}$ for $\textnormal{SM}_i \in~ \mathbb{SM}^{\textnormal{imp}}_{\textnormal{d}_j, \textnormal{s}_u}$ \\

$\mathbb{E}_{\textnormal{d}_j,\textnormal{s}_u}^{\textnormal{exp},t_k}$ & aggregate data of all $\textnormal{E}_i^{\textnormal{exp},t_k}$ for $\textnormal{SM}_i \in~ \mathbb{SM}^{\textnormal{exp}}_{\textnormal{d}_j, \textnormal{s}_u}$ \\



\hline
\end{tabular}
\end{table}



\subsection{Design Requirements}


\subsubsection{Functional Requirements}
\label{Functional requirements} 

\begin{itemize}

\item[(F1)] For each time period $t_k$, each DNO $\textnormal{d}_j$ should access:

\begin{itemize}

\item[a)] $\mathbb{E}^{\textnormal{imp},t_k}_{\textnormal{d}_j}$ and $\mathbb{E}^{\textnormal{exp},t_k}_{\textnormal{d}_j}$, in order to better manage the distribution network in its region,

\item[b)] $\mathbb{E}^{\textnormal{imp},t_k}_{\textnormal{d}_j,\textnormal{s}_u}$ and $\mathbb{E}^{\textnormal{exp},t_k}_{\textnormal{d}_j,\textnormal{s}_u}$, for $u=1, \ldots, {\textnormal{N}_{\textnormal{s}}}$, in order to split the distribution fees fairly among the suppliers.

\end{itemize}

\item[(F2)] For each time period $t_k$, each supplier $\textnormal{s}_u$ should access:

\begin{itemize}

\item[a)] $\mathbb{E}^{\textnormal{imp},t_k}_{\textnormal{s}_u}$ and $\mathbb{E}^{\textnormal{exp},t_k}_{\textnormal{s}_u}$, in order to predict its customers' electricity consumption and production accurately, so that it can avoid receiving imbalance fines,

\item[b)] $\mathbb{E}^{\textnormal{imp},t_k}_{\textnormal{d}_j,\textnormal{s}_u}$ and $\mathbb{E}^{\textnormal{exp},t_k}_{\textnormal{d}_j,\textnormal{s}_u}$ for $j=1, \ldots, {\textnormal{N}_{\textnormal{d}}}$, so it can be assured that it pays the correct transmission and distribution network fees to the TSO and each DNO, respectively. Transmission network fees can also be made region-dependent to encourage suppliers to buy electricity from sources located closer to the demand. 

\end{itemize}

\item[(F3)] For each time period $t_k$, the TSO should access: 

\begin{itemize}

\item[a)] $\mathbb{E}^{\textnormal{imp},t_k}_{\textnormal{d}_j,\textnormal{s}_u}$ and $\mathbb{E}^{\textnormal{exp},t_k}_{\textnormal{d}_j,\textnormal{s}_u}$, for $u=1, \ldots, {\textnormal{N}_{\textnormal{s}}}$, so it can split transmission network fees among suppliers,

\item[b)] $\mathbb{E}^{\textnormal{imp},t_k}_{\textnormal{s}_u}$ and $\mathbb{E}^{\textnormal{exp},t_k}_{\textnormal{s}_u}$, for $u=1, \ldots, {\textnormal{N}_{\textnormal{s}}}$, so it can calculate the imbalance fine for each supplier, 

\item[c)] $\mathbb{E}^{\textnormal{imp},t_k}_{\textnormal{d}_j}$ and $\mathbb{E}^{\textnormal{exp},t_k}_{\textnormal{d}_j}$, for $j=1, \ldots, {\textnormal{N}_{\textnormal{d}}}$, to identify the regions that cause imbalance, thus to decide which measures to take to avoid the imbalance, and 

\item[d)] $\mathbb{E}^{\textnormal{imp},t_k}$ and $\mathbb{E}^{\textnormal{exp},t_k}$, to balance the grid efficiently.

\end{itemize}

\end{itemize}


\subsubsection{Security and Privacy Requirements}
\label{Security requirements}

\begin{itemize}

\item[(S1)] Privacy of users: individual users' fine-grained metering data should not be revealed to any in(ex-)ternal SG entity. 

\item[(S2)] Confidentiality of users' data: only the authorised entities, i.e., TSO, DNO and suppliers, should have access only to the aggregates of users' metering data.

\item[(S3)] Authorisation: SG entities should be allowed to access the aggregate data only of the users whom they provide services to. For the DNO this means only the users living in the region it operates, for the supplier this means only the users who have a contract with it.

\item[(S4)] Fault tolerance: (Partial) loss of metering data should still allow computation of aggregates of the data.

\end{itemize}



\subsection{Cryptographic Notation} 


Security of MPC protocols is typically analysed in the Universally Composable (UC) framework
~\cite{Canetti2000}, under which, the ideal functionality of MPC is modelled as an Arithmetic Black Box (ABB). ABB can be seen as a generic procedure for secure computation, where any party can send its private input to ABB and ask it to calculate any computable function. The functions are represented as arithmetic circuits comprising additions, multiplications, permutations, equality tests, etc. As long as the circuit components are UC secure, the UC framework guarantees that the circuit can be executed securely. Our protocol uses \textit{equality test} and \textit{permutation} as components.



\textit{Equality test} can be implemented in oblivious fashion by using just multiplications and additions. Any existing test~\cite{DKNT06,Toft13} is suitable for use. To simplify it, SMs could share their ID in its bit representation. This way the bit-wise comparison would require only $\sigma$ multiplications, where $\sigma$ is the bit length of the suppliers' ID. Algorithm~\ref{algo:equality_test} illustrates this. It can be optimised by parallelising the computation of multiplications so that only $\log(\sigma)$ communication rounds are needed.


\begin{algorithm}[t]
\scriptsize
\label{algo:equality_test} 
  \SetAlgorithmName{Algorithm}{y}{x} 
  \KwIn{Secret share bit representation of $x$, $[x]_1,\ldots,[x]_{\sigma}$ \\
        $\qquad\quad$Bit representation $y_1,\ldots,y_\sigma$ of public scalar $y$ to which $x$ is compared}
  \KwOut{A secret share of the output of the equality test [c]} 
	$[c] \leftarrow 0;$\\
	\For{$i \leftarrow 1$ \KwTo $\sigma$}	  {
  		$[c'] \leftarrow [x]_{i} +y_{i} -2\cdot ([x]_{i} \cdot y_{i});$\\
        $[c] \leftarrow [c] +[c'] - [c]\cdot [c'];$\\
	}
 \caption{Generic Equality Test}
\end{algorithm}



\textit{Oblivious permutation} can be achieved by using an $n\times n$ Boolean permutation matrix, where $n$ is the size of the input to be permuted. 
Under this approach each entry of the input is multiplied against a corresponding matrix column, and the results are aggregated. This method has $\mathcal{O}(n^2)$ complexity. Other approaches (e.g., use of sorting networks) can achieve better asymptotic complexity. However, using a pre-computed permutation network this can be achieved in (almost) $\mathcal{O}(n\log(n))$ complexity~\cite{CKKL99}. Such approaches have been adapted for practical use in electricity markets~\cite{AV16}.

We assume that all secretly shared values are members of a field $\mathbb{Z}_p$ bounded by a sufficiently large prime $p$, such that no overflow occurs. If fixed point precision is needed, the entries can be multiplied with a large enough constant such that they can be shared as elements of $\mathbb{Z}_{p}$. 






\section{Privacy-preserving Smart Metering Protocol} 
\label{Privacy-preserving Smart Metering Protocols}


In this section we propose a privacy-preserving MPC-based protocol for operational metering data collection. We give an overview of the protocol, and then propose three aggregation algorithms that offer different privacy/performance trade-offs.

\subsection{Overview of our Generic Protocol} 



The generic protocol consists of the following four steps.

\subsubsection{Input data generation and distribution} Each SM generates three data tuples, each containing different shares of the user's contracted suppliers, consumption and generation data, and sends them to the corresponding computational parties.

\subsubsection{Region-based data aggregation} Once the input data of all the SMs are received, the computational parties aggregate the consumption and generation data for each region using one of the three aggregation algorithms described below. The output is in shared form and represents the region-based aggregate consumption and generation data per supplier. 

\subsubsection{Grid-based data aggregation} The computational parties compute the shares of all the grid-based aggregate consumption and generation data by simply adding the corresponding shares of the region-based aggregate data. 

\subsubsection{Output data distribution} Following the functional requirements specified in Section~\ref{Preliminaries}, the shares of the previously calculated aggregations are distributed to the TSO, DNOs and suppliers, accordingly. Finally, these entities reconstruct their required results by reconstructing the corresponding shares.


\begin{algorithm}[t]
\scriptsize
\label{algo:naive_mpc} 
  \SetAlgorithmName{Algorithm}{y}{x} \KwIn{Tuples from region $j$, $\{[\textnormal{s}^{\textnormal{imp}}_u], [\textnormal{s}^{\textnormal{exp}}_u], [ \textnormal{E}^{\textnormal{imp}}_i], [\textnormal{E}^{\textnormal{exp}}_i]\}$ for $\textnormal{SM}_i \in  \mathbb{SM}_{\textnormal{d}_j}$ }
  \KwOut{Shares of aggregate consumption data per supplier, $[\mathbb{E}_{\textnormal{d}_j,\textnormal{s}_u}^{\textnormal{imp}}]$ ~~~~~~~~~~Shares of aggregate production data per supplier, $[\mathbb{E}_{\textnormal{d}_j,\textnormal{s}_u}^{\textnormal{exp}}]$}
  
$[\mathbb{E}_{\textnormal{d}_j,\textnormal{s}_u}^{\textnormal{imp}}] \leftarrow \{0_{1},...,0_{\textnormal{N}_{\textnormal{s}}}\};$ \\

$[\mathbb{E}_{\textnormal{d}_j,\textnormal{s}_u}^{\textnormal{exp}}] \leftarrow \{0_{1},...,0_{\textnormal{N}_{\textnormal{s}}}\};$ \\

	\For{$i \leftarrow 1$ \KwTo $|\mathbb{SM}_{\textnormal{d}_j}|$ }	  {
    	\For{$u \leftarrow 1$ \KwTo ${\textnormal{N}_{\textnormal{s}}}$}
        {
        $[c] \leftarrow [\textnormal{s}^{\textnormal{imp}}_{u}] \stackrel{?}{=} \textnormal{s}_{u};$ \\
        
        $[\mathbb{E}_{\textnormal{d}_j,\textnormal{s}_u}^{\textnormal{imp}}] \leftarrow [\mathbb{E}_{\textnormal{d}_j,\textnormal{s}_u}^{\textnormal{imp}}] + [c] * [\textnormal{E}^{\textnormal{imp}}_i];$ \\
        }
            	
                
                \For{$u \leftarrow 1$ \KwTo ${\textnormal{N}_{\textnormal{s}}} $  } 
       {
        $[c]\leftarrow [\textnormal{s}^{\textnormal{exp}}_{u}] \stackrel{?}{=} \textnormal{s}_{u};$ \\
        
        $[\mathbb{E}_{\textnormal{d}_j,\textnormal{s}_u}^{\textnormal{exp}}] \leftarrow [\mathbb{E}_{\textnormal{d}_j,\textnormal{s}_u}^{\textnormal{exp}}] + [c] * [ \textnormal{E}^{\textnormal{exp}}_i];$ \\
       } 
	}         
 \caption{Na\"{\i}ve Aggregation Algorithm (NAA)}
\end{algorithm}


\subsection{Region-based Data Aggregation Algorithms}

In this section, we present our three region-based data aggregation algorithms that offer different trade-offs in terms of security, flexibility and performance. 

\subsubsection{Na\"{\i}ve Aggregation Algorithm (NAA)}


A na\"{\i}ve approach to perform data aggregation with perfect privacy would be to implement a basic circuit that uses equality tests to identify users' suppliers. As shown in Algorithm~\ref{algo:naive_mpc}, SMs send their tuples $\{[\textnormal{s}^{\textnormal{imp}}_u], [\textnormal{s}^{\textnormal{exp}}_u], [ \textnormal{E}^{\textnormal{imp}}_i], [\textnormal{E}^{\textnormal{exp}}_i]\}$ to the DCC servers, so that the servers can classify the inputs by using oblivious comparisons. Although the algorithm is fairly adaptive to a growing number of suppliers, denoted as $\textnormal{N}_{\textnormal{s}}$, it is expensive in terms of performance as it still requires $\mathcal{O}(|\mathbb{SM}_{\textnormal{d}_j}|\cdot{\textnormal{N}_{\textnormal{s}}})$ equality tests, where $|\mathbb{SM}_{\textnormal{d}_j}|$ is the number of SMs in a given region $j$. 


\begin{algorithm}[t]
\scriptsize
\label{algo:ncp_MPC} 
  \KwIn{Tuples from region $j$, $\{[\textnormal{s}^{\textnormal{imp}}_u], [\textnormal{s}^{\textnormal{exp}}_u], [ \textnormal{E}^{\textnormal{imp}}_i], [\textnormal{E}^{\textnormal{exp}}_i]\}$ for $\textnormal{SM}_i \in  \mathbb{SM}_{\textnormal{d}_j}$ }
  \KwOut{Shares of aggregate consumption data per supplier, $[\mathbb{E}_{\textnormal{d}_j,\textnormal{s}_u}^{\textnormal{imp}}]$ ~~~~~~~~~~Shares of aggregate production data per supplier, $[\mathbb{E}_{\textnormal{d}_j,\textnormal{s}_u}^{\textnormal{exp}}]$}
  
$[\mathbb{E}_{\textnormal{d}_j,\textnormal{s}_u}^{\textnormal{imp}}] \leftarrow \{0_{1},...,0_{\textnormal{N}_{\textnormal{s}}}\};$ \\

$[\mathbb{E}_{\textnormal{d}_j,\textnormal{s}_u}^{\textnormal{exp}}] \leftarrow \{0_{1},...,0_{\textnormal{N}_{\textnormal{s}}}\};$ \\

$[\mathbb{SM}_{\textnormal{d}_j}'] \leftarrow \mathtt{permute}([\mathbb{SM}_{\textnormal{d}_j}]);$ \\

	\For{$i \leftarrow 1$ \KwTo $|\mathbb{SM}_{\textnormal{d}_j}'|$ }	  {
    	$\textnormal{s}^{\textnormal{imp}}_u \leftarrow \mathtt{open}([\textnormal{s}^{\textnormal{imp}}_u]);$\\
    	\For{$u \leftarrow 1$ \KwTo $\textnormal{N}_{\textnormal{s}}$}
        {
        $c\leftarrow \textnormal{s}^{\textnormal{imp}}_u == \textnormal{s}_{u};$ \\
        
                $[\mathbb{E}_{\textnormal{d}_j,\textnormal{s}_u}^{\textnormal{imp}}] \leftarrow [\mathbb{E}_{\textnormal{d}_j,\textnormal{s}_u}^{\textnormal{imp}}] + c * [ \textnormal{E}^{\textnormal{imp}}_i];$ \\
                
        }
        
	}
    
    $[\mathbb{SM}_{\textnormal{d}_j}'] \leftarrow \mathtt{permute}([\mathbb{SM}_{\textnormal{d}_j}]); $
    \\
    \For{$i \leftarrow 1$ \KwTo $|\mathbb{SM}_{\textnormal{d}_j}'|$ }	  {
    	$\textnormal{s}^{\textnormal{exp}}_u \leftarrow \mathtt{open}([\textnormal{s}^{\textnormal{exp}}_u]);$\\
    	\For{$u \leftarrow 1$ \KwTo $\textnormal{N}_{\textnormal{s}}$}
        {
        $c\leftarrow \textnormal{s}^{\textnormal{exp}}_u == \textnormal{s}_{u};$ \\
        
                $[\mathbb{E}_{\textnormal{d}_j,\textnormal{s}_u}^{\textnormal{exp}}] \leftarrow [\mathbb{E}_{\textnormal{d}_j,\textnormal{s}_u}^{\textnormal{exp}}] + c * [ \textnormal{E}^{\textnormal{exp}}_i];$ \\
                
        }
        
	}
 \caption{No Comparison Aggr. Algorithm (NCAA)} 
\end{algorithm}


\subsubsection{No Comparison Aggregation Algorithm (NCAA)} 

To improve the performance of the aggregation algorithm, some level of disclosure to the DCC servers can be allowed, in this case, the number of users linked to each supplier. As shown in Algorithm~\ref{algo:ncp_MPC}, the DCC servers permute the tuples corresponding to the same region and aggregate them in a non-interactive way afterwards. Considering that its complexity is dominated by the oblivious permutation calls, NCAA multiplication bound is $\mathcal{O}(|\mathbb{SM}_{\textnormal{d}_j}|\cdot \log(|\mathbb{SM}_{\textnormal{d}_j}|)$. Also, NCAA keeps its flexibility with respect to $\textnormal{N}_{\textnormal{s}}$ at the cost of disclosing the number of SMs associated to each supplier. 


\begin{algorithm}[b]
\scriptsize
\label{algo:nip_mpc} 
  \SetAlgorithmName{Algorithm}{y}{x} \KwIn{Tuples from region $j$, $\{[ \mathbf{E}^{\textnormal{imp}}_i], [\mathbf{E}^{\textnormal{exp}}_i]\}$ for $\textnormal{SM}_i \in  \mathbb{SM}_{\textnormal{d}_j}$, where $ \mathbf{E}^{\textnormal{imp}}_i$ and $ \mathbf{E}^{\textnormal{exp}}_i$ are vectors of size $\textnormal{N}_{\textnormal{s}}$ with only one non-zero entry at position $u$ }
  
  \KwOut{Shares of aggregate consumption data per supplier, $[\mathbb{E}_{\textnormal{d}_j,\textnormal{s}_u}^{\textnormal{imp}}]$ ~~~~~~~~~~Shares of aggregate production data per supplier, $[\mathbb{E}_{\textnormal{d}_j,\textnormal{s}_u}^{\textnormal{exp}}]$}
  
$[\mathbb{E}_{\textnormal{d}_j,\textnormal{s}_u}^{\textnormal{imp}}] \leftarrow \{0_{1},...,0_{\textnormal{N}_{\textnormal{s}}}\};$\\ 
$[\mathbb{E}_{\textnormal{d}_j,\textnormal{s}_u}^{\textnormal{exp}}] \leftarrow \{0_{1},...,0_{\textnormal{N}_{\textnormal{s}}}\};$\\
	\For{$i \leftarrow 1$ \KwTo $|\mathbb{SM}_{\textnormal{d}_j}|$}	  {
    	\For{$u \leftarrow 1$ \KwTo $\textnormal{N}_{\textnormal{s}}$}
        {
        $[\mathbb{E}_{\textnormal{d}_j,\textnormal{s}_u}^{\textnormal{imp}}] \leftarrow [\mathbb{E}_{\textnormal{d}_j,\textnormal{s}_u}^{\textnormal{imp}}] +  [\textnormal{E}^{\textnormal{imp}}_{i,u}];$\\
        $[\mathbb{E}_{\textnormal{d}_j,\textnormal{s}_u}^{\textnormal{exp}}] \leftarrow [\mathbb{E}_{\textnormal{d}_j,\textnormal{s}_u}^{\textnormal{exp}}] +  [\textnormal{E}^{\textnormal{exp}}_{i,u}];$\\
        }
	}
 \caption{Non-Interactive Aggr. Algorithm (NIAA)}
\end{algorithm}


\subsubsection{Non-Interactive Aggregation Algorithm (NIAA)} 

To further improve the performance of the aggregation algorithm, the input data of SMs can be tweaked such that the aggregation could be done without the need of communication between the DCC servers. To achieve this, SMs have to encode their input data into vectors of all zeros but one unique non-zero entry. 
These vectors are of size $\textnormal{N}_{\textnormal{s}}$ and the non-zero entries are their $\textnormal{E}^{\textnormal{imp}}$ and $\textnormal{E}^{\textnormal{exp}}$, respectively. This way the DCC servers only need to process the aggregation of the shares, which is non-interactive for any generalized Linear Secret Sharing Scheme (LSSS). By reducing the flexibility ($\textnormal{N}_{\textnormal{s}}$ has to be fixed), NIAA, as shown in Algorithm~\ref{algo:nip_mpc}, is implemented with neither comparison nor multiplication operations. To support the addition of a new supplier, SMs will have to use a vector with a sufficiently large pre-fixed size, providing $0$ for the non-used slots, so that the system is flexible in accommodating a large number of suppliers. An easy alternative would be to allow the system to feed (via an update) all the SMs with a parameter -- the number of suppliers -- so that SMs will encode their inputs as vectors of correct length. Moreover, the supplier ID position has to be agreed in advance. NIAA also produces no leakage, hence it achieves perfect security. 


\section{Security Analysis}
\label{Security Analysis}

We analyse the security of our protocol in the Universal Composability (UC) framework~\cite{Canetti2000}, in which the ideal functionality of MPC is modelled as Arithmetic Black Box (ABB)~\cite{DN03}. ABB is a generic procedure for MPC providing an abstraction of the details of MPC operations and of secret sharing. Players can send their private input to ABB to compute any computable function on their private inputs. The results from the computation are stored in the internal state of ABB to be used in the subsequent computations. Stored values are only made public if enough number of players agree. Formally, ABB is defined as follows.
\begin{definition}[ABB Functionality $\mathcal{F}_{\textit{ABB}}$]
\label{def:FABB}
The ideal functionality $\mathcal{F}_\textit{ABB}$ for MPC is defined as follows:
\begin{itemize}
\item $\mathtt{input}$: Receive a value $\alpha \in \mathbb{Z}_M$ and store $\alpha$. 
\item $\mathtt{share}(\alpha)$: Create a share $[\alpha]$ of $\alpha$.
\item $\mathtt{product}([\alpha], [\beta])$: Compute $\gamma = \alpha \times \beta$ and store $[\gamma]$.
\item $\mathtt{compare}([\alpha],[\beta])$: Compare $\alpha$ and $\beta$, and return $0$ if $\alpha<\beta$ and 1 otherwise.
\item $\mathtt{equal}([\alpha],[\beta])$: Check if $\alpha=\beta$; return $1$ if $\alpha=\beta$, 0 otherwise.
\item $\mathtt{permute([X])}$: Given an input $[X]\in\mathbb{Z}_M^n$ return a random permutation of it. 
\item $\mathtt{open}([\alpha])$: Send the value $\alpha$ to all players.
\end{itemize}
Addition and scalar multiplication are denoted by  their corresponding conventional symbols $+$ and $\times$.
\end{definition}


\begin{definition}[UC-security \cite{C00}]
\label{def:UC-security}
 A real protocol $\pi$ is UC-secure if, for any adversary $\mathcal{A}$, there exists a simulator $\mathcal{S}$ for which no environment $\mathcal{Z}$ can distinguish with a non-negligible probability if it is interacting with $\mathcal{A}$ and $\pi$ or $\textsf{S}$ and the ideal functionality $\mathcal{F}$.
\end{definition}

\begin{definition}[Universal Composition \cite{C00}]
\label{UC-theorem}
Let $\pi$ and $\rho$ be two protocols such that $\rho$
$\epsilon_1$-UC-emulates $\mathcal{G}$ and $\pi\circ\mathcal{G}$ $\epsilon_2$-UC-emulates
$\mathcal{F}$, when using $\mathcal{G}$ as a subroutine. Then
$\pi\circ\rho$ $(\epsilon_1+\epsilon_2)$-UC-emulates $\mathcal{F}$, when using
$\rho$ as a subroutine.
\end{definition}


We define the ideal functionality of the operation metering data collection protocol as follows.  
\begin{definition}[Ideal Functionality of Operational Metering Data Collection Protocol $\mathcal{F}_\textit{OMDCP}$]
Given a set of input tuples $\{\textnormal{s}^{\textnormal{imp}}_u, \textnormal{s}^{\textnormal{exp}}_u,  \textnormal{E}^{\textnormal{imp}}_i, \textnormal{E}^{\textnormal{exp}}_i\}$ from SMs, the ideal functionality $\mathcal{F}_{\textit{OMDCP}}$ computes the aggregate consumption data $\mathbb{E}_{\textnormal{d}_j,\textnormal{s}_u}^{\textnormal{imp}}$ and aggregate production data $\mathbb{E}_{\textnormal{d}_j,\textnormal{s}_u}^{\textnormal{exp}}$ per supplier, and returns the results to the relevant parties (i.e., suppliers, DNOs, and TSO). 
\end{definition}

Next, we show that the three region-based data aggregation algorithms -- NAA, NCAA, and NIAA -- securely realise the ideal functionality $\mathcal{F}_\textit{OMDCP}$. This follows from the fact that all three algorithms use MPC operations provided by ABB, which are both sequentially and concurrently universal composable. 

\begin{theorem}
Let $\pi_{\text{NAA}}$ be the operational metering data collection protocol employing the NAA algorithm. Then $\pi_\text{NAA}$ securely emulates $\mathcal{F}_{\textit{OMDCP}}$.
\end{theorem}
\begin{proof} 
The proof is a straightforward application of the Universal Composition Theorem~\ref{UC-theorem}, since (i) each ABB operation is universally composable and (ii) the NAA algorithm comprises only ABB operations, namely, $\mathtt{equal}$, addition, and scalar multiplication.  
\end{proof}

\begin{theorem}
Let $\pi_{\text{NCAA}}$ be the operational metering data collection protocol employing the NAA algorithm. Then $\pi_\text{NCAA}$ securely emulates $\mathcal{F}_{\textit{OMDCP}}$.
\end{theorem}

\begin{proof} 
For the same reason as in the proof of the previous Theorem. In this case, the ABB operations that are used in NCAA are $\mathtt{permute}$, $\mathtt{equal}$, addition, scalar multiplication, and $\mathtt{open}$.  
\end{proof}

\begin{theorem}
Let $\pi_{\text{NIAA}}$ be the operational metering data collection protocol employing the NIAA algorithm. Then $\pi_\text{NAA}$ securely emulates $\mathcal{F}_{\textit{OMDCP}}$.
\end{theorem}

\begin{proof} 
In the case of NIAA, only addition is used, so the proof is again straightforward.   
\end{proof}

Note that we build all our protocols over the idealised functionality given by $\mathcal{F}_\textit{ABB}$, and no intermediate value is revealed on any of them. Security then trivially follows from the simulation of the protocols where the adversary is restricted to what already is allowed under the ideal functionality. This principle applies to all the results in this paper.

\textit{Fault tolerance and collusion-attacks resistance:}
Our solutions are also fault-tolerant. This property is guaranteed due to the fact that in all our solutions the SMs send shares of the consumption data to the respective entities. These shares are generated using Shamir secret sharing with a threshold, say, $\tau$. Therefore, even if some shares of the metering data are lost in transit, as long as $\tau+1$ shares are present, the DCC servers can still compute the aggregates correctly.

Regarding security against collusion attacks, security under MPC can be achieved under two possible settings: an honest majority, and a dishonest majority e.g.,~\cite{CCD88,BGW88,DPSZ12,DKLPSS13}. In the former, a majority of computational parties (i.e., at least two of the DCC servers in our case) need to remain honest in order to achieve perfect security. In the latter case, we can protect against any collusion attacks as long as there is at least one honest computational party, by choosing a protocol secure against dishonest majority to implement the underlying MPC functionality.

\section{Performance Evaluation}
\label{Performance Evaluation}

This section evaluates the performance of our protocol (and our proposed data aggregation algorithms) in terms of computational complexity and communication cost using parameters of the smart metering architecture in the UK. In addition, we compare our protocol to (i) the traditional protocol (denoted as TRAD) proposed by the UK government, and (ii) the DEP2SA protocol~\cite{DEP2SA} which uses the same system model as this work. Note that TRAD does not provide sufficient privacy protection as the DCC access all metering data of all users.

\begin{table*}[!t]
  \centering  \caption{The computational complexity of our protocol when a different data aggregation algorithm is used.}
  \scalebox{0.9}{
  \begin{tabular}{l c c c c c}
  \cline{2-6}
  & \textbf{SM} & \textbf{DCC servers} & \textbf{TSO}  & \textbf{DNO} & \textbf{Supplier} \\
    
  \hline

\textbf{TRAD protocol}   & - & - & - & - & - \\

\textbf{DEP2SA protocol~\cite{DEP2SA}}  &  1 enc & - & - & $\textnormal{N}_{\textnormal{s}}$ (dec + rnr) & $\textnormal{N}_{\textnormal{d}}$ enc \\

  \textbf{Our protocol with NAA}    &  1 share & $|\textnormal{s}_{u}|\times |\mathbb{SM}_{\textnormal{d}_j}|\times \textnormal{N}_{\textnormal{s}} + |\mathbb{SM}_{\textnormal{d}_j}|\times \textnormal{N}_{\textnormal{s}}$ multiplication & $\textnormal{N}_{\textnormal{d}}~\times~\textnormal{N}_{\textnormal{s}}$ open & $\textnormal{N}_{\textnormal{s}}$ open  & $\textnormal{N}_{\textnormal{d}}$ open \\
  
  \textbf{Our protocol with NCAA}   &  1 share & $2\times(|\mathbb{SM}_{\textnormal{d}_j}| \times \log(|\mathbb{SM}_{\textnormal{d}_j}|) +|\mathbb{SM}_{\textnormal{d}_j}|$ multiplication & $\textnormal{N}_{\textnormal{d}}~\times~\textnormal{N}_{\textnormal{s}}$ open & $\textnormal{N}_{\textnormal{s}}$ open  & $\textnormal{N}_{\textnormal{d}}$ open \\
   
  \textbf{Our protocol with NIAA}   &  $\textnormal{N}_{\textnormal{s}}$ share  & 0 multiplication & $\textnormal{N}_{\textnormal{d}}~\times~\textnormal{N}_{\textnormal{s}}$ open & $\textnormal{N}_{\textnormal{s}}$ open  & $\textnormal{N}_{\textnormal{d}}$ open \\
  
  \hline

 \hline
  \end{tabular}}
\label{tab:computational complexity}
\end{table*}



\begin{table}[!t]
  \centering  \caption{Simulation parameters.}
  \scalebox{1.0}{
  \begin{tabular}{l l}
  \hline
   \textbf{Parameter} & \textbf{Value} \\
  \hline
$\textnormal{N}_{\textnormal{d}}$  - number of DNOs &  $14$ \\
$\textnormal{N}_{\textnormal{s}}$ - number of suppliers & $ 10$ \\
$|\textnormal{s}_{u}|$ - length of supplier ID & $ 8$ \\ 
$|\mathbb{SM}_{\textnormal{d}_j}|$   - number of SMs in each DNO &  $0.5M - 4M$ \\
 $|x|, |r|$ - length of data, random number & $ 32 $ \\
 $|[x]|$ - length of share & $63 $ \\
 $|{c}|$ - length of encrypted data with symmetric key & $ 128$ \\
 $|{C}|$  - length of encrypted data with public key & $1024$ \\
  \hline
  \end{tabular}}
\label{tab:simulation_parameters}
\end{table}



\subsection{Computational Complexity}

The most computationally demanding step of our protocol is the \textit{region-based aggregation} algorithm. Therefore, we focus on this step. Moreover, since the cost of a share, addition and open operations is negligible compared to the cost of a multiplication operation (in an MPC setting), we take into account only the number of multiplications in our calculation.

\subsubsection{NAA complexity} This algorithm contains two loops which have the same number of multiplications. For each loop, NAA requires $|\textnormal{s}_{u}|\times |\mathbb{SM}_{\textnormal{d}_j}|\times \textnormal{N}_{\textnormal{s}}$ multiplications to perform the equality tests needed, and $|\mathbb{SM}_{\textnormal{d}_j}|\times \textnormal{N}_{\textnormal{s}}$ multiplications needed for the aggregation, where $|\textnormal{s}_{u}|$ is the bit length of the supplier ID, $|\mathbb{SM}_{\textnormal{d}_j}|$ is the number of SMs per region and $\textnormal{N}_{\textnormal{s}}$ is the number of suppliers in the retail market. However, as both loops are parallelizible, the total number of multiplications in NAA is equal to $|\textnormal{s}_{u}|\times |\mathbb{SM}_{\textnormal{d}_j}|\times \textnormal{N}_{\textnormal{s}} + |\mathbb{SM}_{\textnormal{d}_j}|\times \textnormal{N}_{\textnormal{s}}$. 

\subsubsection{NCAA complexity} The number of multiplications used by the NCAA depends on the permutation network used. For instance, the Batcher odd–even merge sorting network requires $|\mathbb{SM}_{\textnormal{d}_j}| \times \log^2(|\mathbb{SM}_{\textnormal{d}_j}|)$ exchange gates. Each of these gates requires three multiplications per item being permuted, in this case the supplier ID and the respective electricity consumption or generation value. Also, the open operation performed by the DCC servers has the same computational cost as performing a multiplication. In total, this adds up to $2\times(|\mathbb{SM}_{\textnormal{d}_j}| \times \log^2(|\mathbb{SM}_{\textnormal{d}_j}|) +|\mathbb{SM}_{\textnormal{d}_j}|$ multiplication-equivalent operations per loop. However, a permutation network can be built with only $|\mathbb{SM}_{\textnormal{d}_j}| \times \log(|\mathbb{SM}_{\textnormal{d}_j}|)$ exchange gates~\cite{CKKL99}, reducing the total to $2\times(|\mathbb{SM}_{\textnormal{d}_j}| \times \log(|\mathbb{SM}_{\textnormal{d}_j}|) +|\mathbb{SM}_{\textnormal{d}_j}|$. 

\subsubsection{NIAA complexity} NIAA does not perform any multiplications. As the cost of aggregation is negligible, given that it is just an arithmetic aggregation of shares, the total computational complexity of NIAA is negligible.

Table~\ref{tab:computational complexity} summarises the computational complexity of our data aggregation algorithms on per entity base. The cost of the operations performed by each SM, TSO, DNO and supplier is negligible compared to the cost of the operations performed by the DCC servers. In terms of computational complexity, NIAA is the most efficient aggregation algorithm as it does not require any communication between the DCC servers.

Note that TRAD does not introduce any additional cost at each entity as it does not provide privacy protection. In the case of DEP2SA, each SM encrypts its data using homomorphic encryption algorithm, each DNO performs $\textnormal{N}_{\textnormal{s}}$ decryption as well as recovery of a random number from each of the ciphertexts it receives, and each supplier has to perform $\textnormal{N}_{\textnormal{d}}$ encryptions. Similarly to the case of TRAD, DCC does not perform any computationally expensive operations.



%




%



\begin{figure}[t]
\centering
\includegraphics[trim= 0 0 0 0,clip=true,width=2.39in]{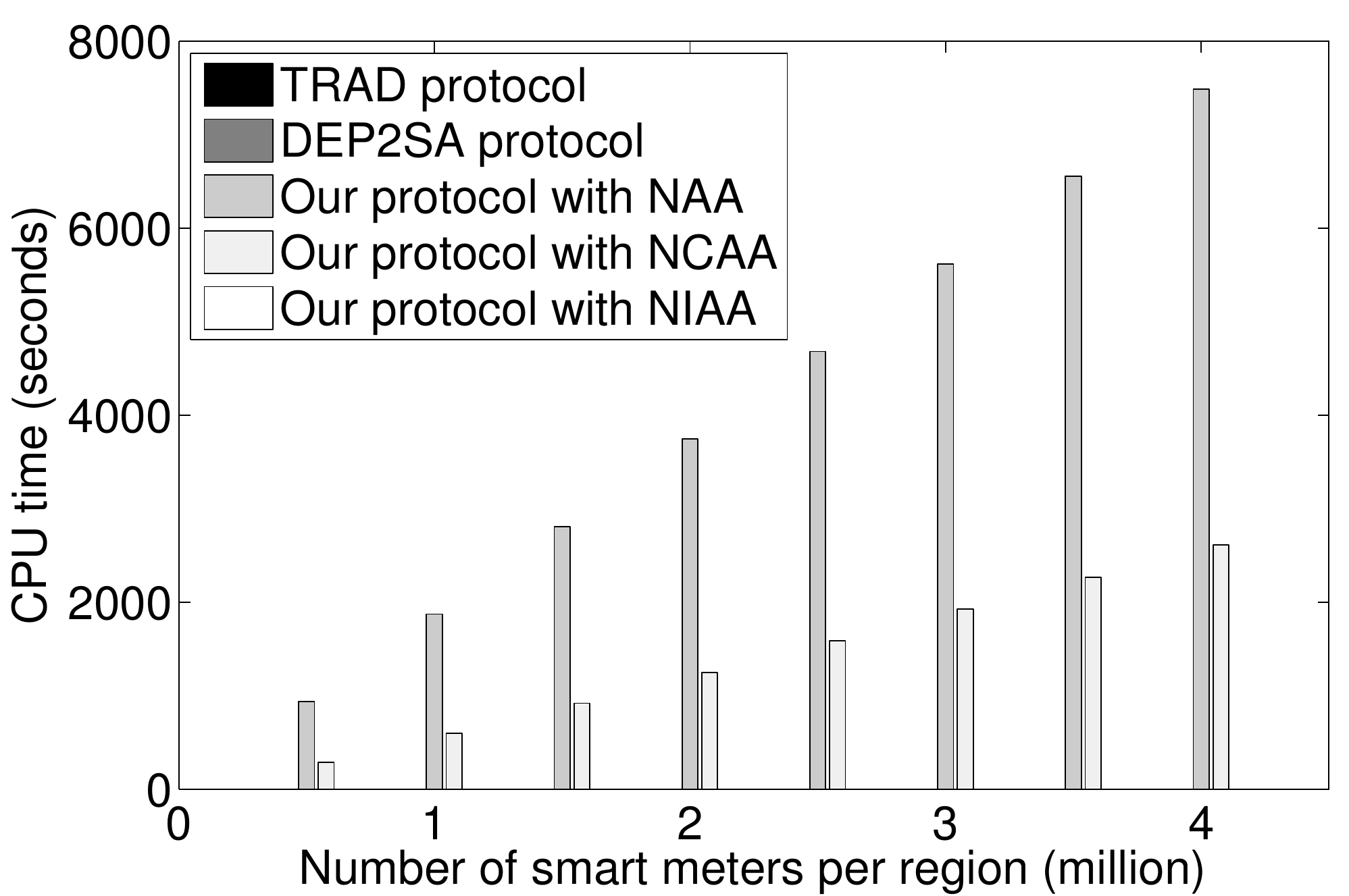}
\caption{Computational cost of the DCC servers.}
\label{fig:computational_results}
\end{figure}


We also conducted an experiment to test the performance of our algorithms. We used C++ and custom implementations of Shamir's SSS~\cite{Shamir79}, its linear addition and improved BGW protocol from Gennaro~et~al.~\cite{GRR90}, all presented in~\cite{Aly15}. We made use of the generalized equality test from Algorithm~\ref{algo:equality_test}. We run the three computational parties on the same machine, a 64-bit 2*2*10-cores Intel Xeon E5-2687 server at 3.1GHz, thus our results do not consider  network latency.


We first executed 2 million multiplications which, on average, resulted in $20.8\times 10^{-6}$ seconds per multiplication. We then calculated the CPU time needed by our algorithms for various settings. For our calculations we used the following parameters based on the UK's electrical grid~\cite{Elexon} and smart metering architecture~\cite{UK_DCC}: $\textnormal{N}_{\textnormal{d}} = 14$, $\textnormal{N}_{\textnormal{s}} = 10$, 
$|\textnormal{s}_{u}| = 8$, and
$|\mathbb{SM}_{\textnormal{d}_j}| = \{0.5M, \ldots, 4M\}$. All of our simulation parameters are given in Table~\ref{tab:simulation_parameters}. Note that the computational complexity does not depend on the metering data but on the smart metering architecture. Figure~\ref{fig:computational_results} depicts our experimental results. They indicate all the necessary CPU time required regardless of the number of processors. Considering that in each UK region there will be on average 2.2 million SMs, our protocol could be executed in less than ten minutes, even if NAA (our most computationally demanding algorithm) is used, by simply dividing the work between eight threads, thus making it practical for the UK smart metering architecture.

\begin{table*}[h]
  \centering  \caption{The communication overhead of the traditional protocol and our protocol.}
  \scalebox{0.9}{
  \begin{tabular}{l c c c}
  \cline{2-4}
  & \textbf{SMs-to-DCC} & \textbf{Between-DCC} & \textbf{DCC-to-TSO/DNOs/Suppliers}  \\ 
  \hline

\textbf{TRAD protocol}   &  $2\times \textnormal{N}_{\textnormal{d}} \times |\mathbb{SM}_{\textnormal{d}_j}| \times |x|$ & 0 & $ 6 \times \textnormal{N}_{\textnormal{d}} \times \textnormal{N}_{\textnormal{s}} \times |x|$ \\
  
 
\textbf{DEP2SA protocol~\cite{DEP2SA}}  &  $2\times \textnormal{N}_{\textnormal{d}} \times |\mathbb{SM}_{\textnormal{d}_j}| \times |C|$ & 0 & $ 2 \times \textnormal{N}_{\textnormal{d}} \times \textnormal{N}_{\textnormal{s}} \times (2 \times |C| + |x| + |r|)$ \\

  \textbf{Our protocol with NAA}    &  $12\times \textnormal{N}_{\textnormal{d}} \times |\mathbb{SM}_{\textnormal{d}_j}| \times |[x]|$ & $6\times |[x]| \times (|\textnormal{s}_{u}|\times |\mathbb{SM}_{\textnormal{d}_j}|\times \textnormal{N}_{\textnormal{s}} + |\mathbb{SM}_{\textnormal{d}_j}|\times \textnormal{N}_{\textnormal{s}})$ &  $ 18 \times \textnormal{N}_{\textnormal{d}} \times \textnormal{N}_{\textnormal{s}} \times |[x]|$ \\
  
  \textbf{Our protocol with NCAA}   &  $12\times \textnormal{N}_{\textnormal{d}} \times |\mathbb{SM}_{\textnormal{d}_j}| \times |[x]|$ & $ 6\times |[x]| \times (2\times(|\mathbb{SM}_{\textnormal{d}_j}| \times \log(|\mathbb{SM}_{\textnormal{d}_j}|) +|\mathbb{SM}_{\textnormal{d}_j}|)$ &  $ 18 \times \textnormal{N}_{\textnormal{d}} \times \textnormal{N}_{\textnormal{s}} \times |[x]|$ \\
   
  \textbf{Our protocol with NIAA}   &  $6\times \textnormal{N}_{\textnormal{d}} \times |\mathbb{SM}_{\textnormal{d}_j}| \times \textnormal{N}_{\textnormal{s}} \times |[x]|$ & 0 & $ 18 \times \textnormal{N}_{\textnormal{d}} \times \textnormal{N}_{\textnormal{s}} \times |[x]|$ \\
  \hline
  
  \end{tabular}}
\label{tab:communication_overhead}
\end{table*}



\begin{figure*}[t]
  \centering

\subfloat[At the SMs-to-DCC part]{\includegraphics[trim = 0 0 0 0,clip=true, width=2.24in]{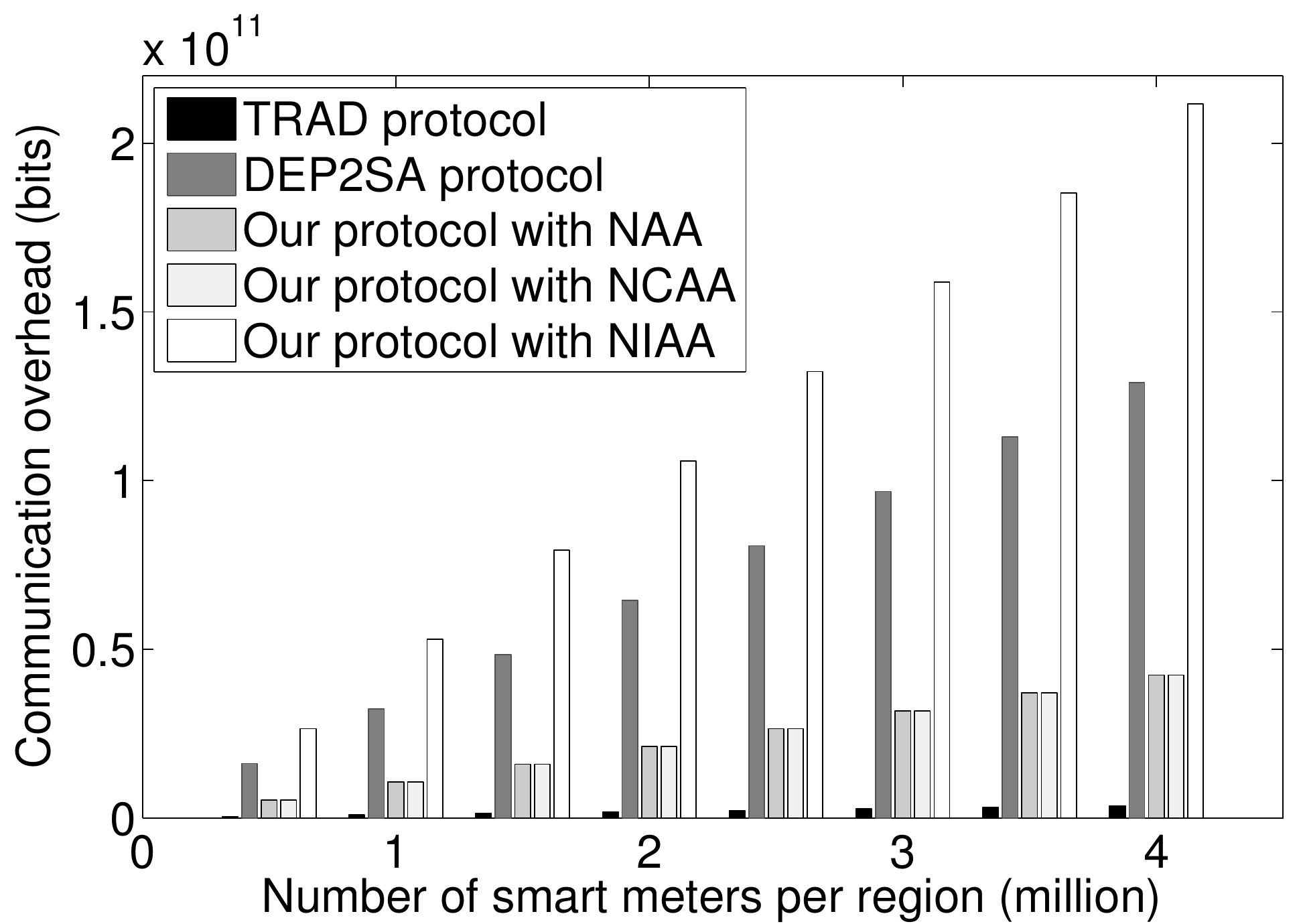}
  \label{fig:sfig1}}
  \hfil
\subfloat[At the Between-DCC part]{\includegraphics[trim = 0 0 0 0, clip=true, width=2.24in]{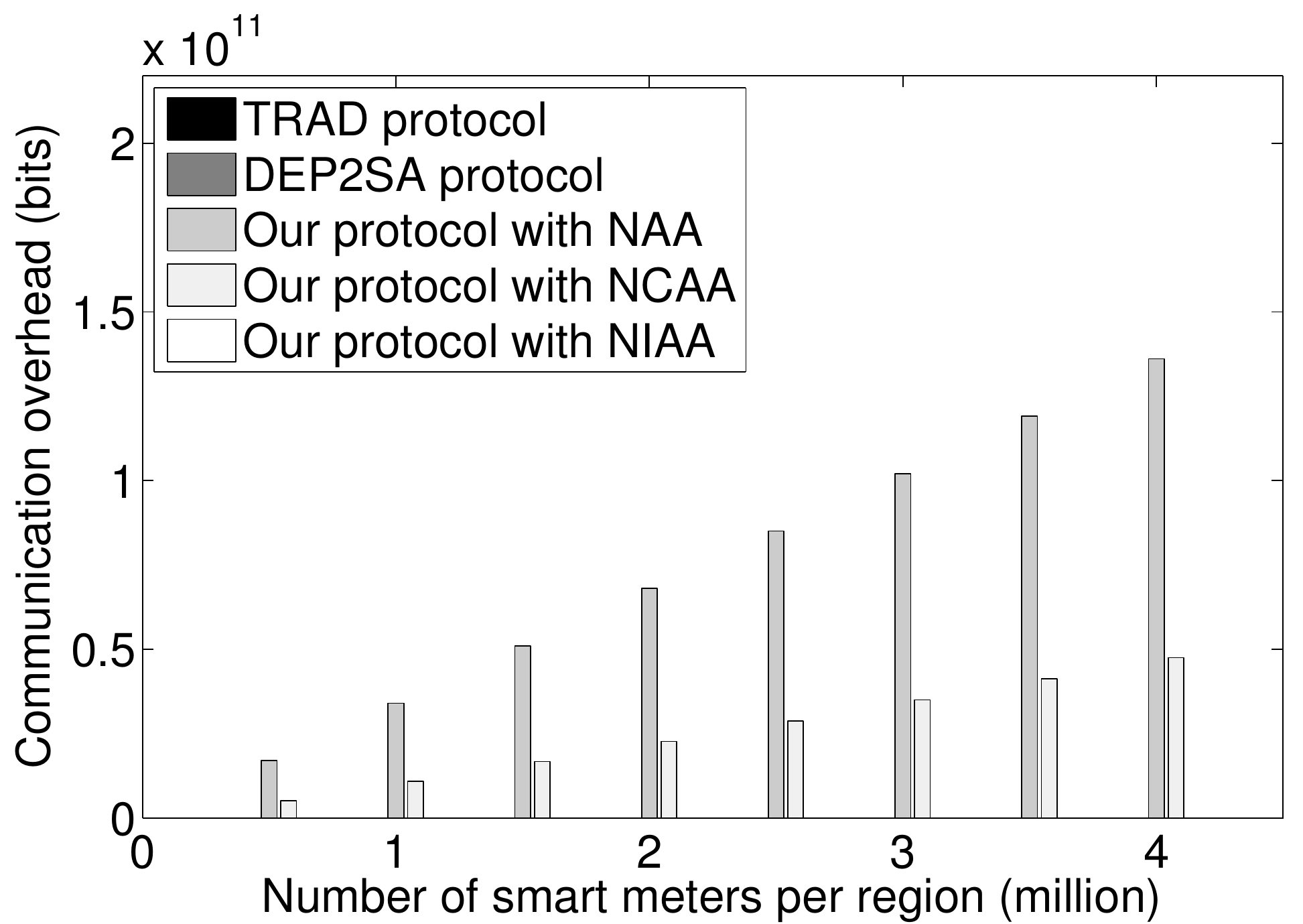}
  \label{fig:sfig2}}
  \hfil
\subfloat[At the DCC-to-TSO/DNOs/Suppliers part]{\includegraphics[trim = 0 -14 0 0,clip=true, width=2.24in]{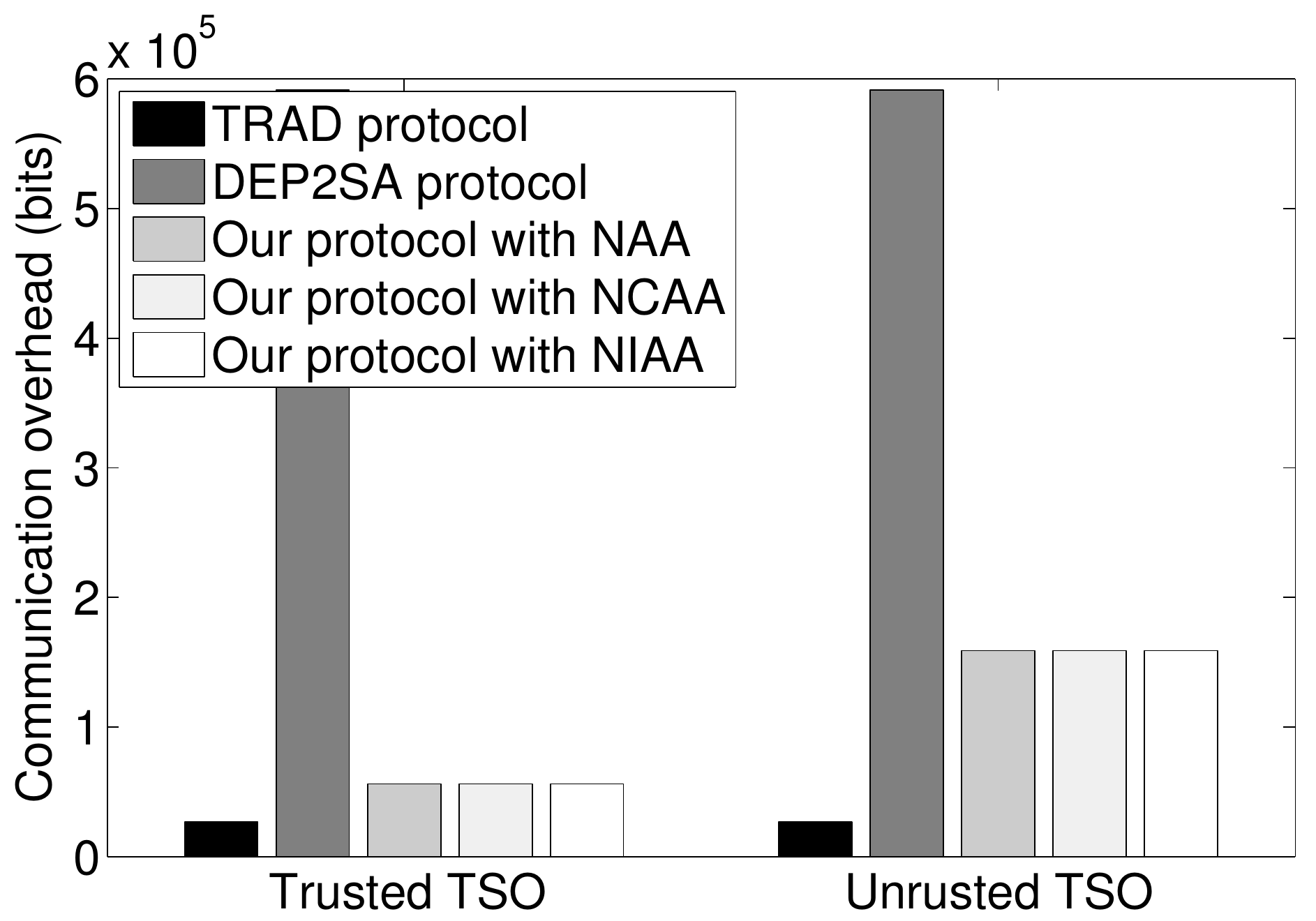}
  \label{fig:sfig3}}
  \hfil
  \caption{The communication overhead of our protocol at different parts of the grid.}
\label{fig:communication_overhead}
\end{figure*}



\subsection{Communication Cost} 

The communication cost of our protocol can be divided in three parts: SMs-to-DCC, Between-DCC and DCC-to-TSO/DNOs/Suppliers. For each part, we evaluate the communication cost of our protocol when a different aggregation algorithm is used and compare it to TRAD and DEP2SA.

\subsubsection{SMs-to-DCC part} In each time slot each SM sends its tuple to each of the DCC servers. If our protocol uses NAA or NCAA, the format of the tuple is $\{[\textnormal{s}^{\textnormal{imp}}_u], [\textnormal{s}^{\textnormal{exp}}_u], [ \textnormal{E}^{\textnormal{imp}}_i], [\textnormal{E}^{\textnormal{exp}}_i]\}$. Assuming there are three DCC servers, the communication cost is $3\times \textnormal{N}_{\textnormal{d}} \times |\mathbb{SM}_{\textnormal{d}_j}| \times ([\textnormal{s}^{\textnormal{imp}}_u] +  [\textnormal{s}^{\textnormal{exp}}_u] + [ \textnormal{E}^{\textnormal{imp}}_i] + [\textnormal{E}^{\textnormal{exp}}_i])$. If our protocol uses NIAA, the tuple's format is $\{[ \mathbf{E}^{\textnormal{imp}}_i], [\mathbf{E}^{\textnormal{exp}}_i]\}$, where $\{[ \mathbf{E}^{\textnormal{imp}}_i], [\mathbf{E}^{\textnormal{exp}}_i]\}$ are shares of vectors with size $\textnormal{N}_{\textnormal{s}}$. This adds up to a cost of $3\times \textnormal{N}_{\textnormal{d}} \times \textnormal{N}_{\textnormal{s}} \times |\mathbb{SM}_{\textnormal{d}_j}| \times ([ \textnormal{E}^{\textnormal{imp}}_i] + [\textnormal{E}^{\textnormal{exp}}_i])$. If TRAD is used, each SM sends $\{  \textnormal{E}^{\textnormal{imp}}_i, \textnormal{E}^{\textnormal{exp}}_i\}$ to the DCC which is a single entity in this case. This results in a communication cost of $\textnormal{N}_{\textnormal{d}} \times |\mathbb{SM}_{\textnormal{d}_j}| \times ( \textnormal{E}^{\textnormal{imp}}_i + \textnormal{E}^{\textnormal{exp}}_i)$. If DEP2SA is used, the communication cost is $\textnormal{N}_{\textnormal{d}} \times |\mathbb{SM}_{\textnormal{d}_j}| \times ( \textnormal{C}^{\textnormal{imp}} + \textnormal{C}^{\textnormal{exp}})$, where $\textnormal{C}^{\textnormal{imp}}$ and $\textnormal{C}^{\textnormal{exp}}$ are the ciphertext of $\textnormal{E}^{\textnormal{imp}}_i$ and $\textnormal{E}^{\textnormal{exp}}_i$, respectively. Note that these ciphertexts are the results of the homomorphic encryption operation performed by the SMs.


\subsubsection{Between-DCC part} In each time slot the DCC servers need to communicate between each other in order to preform the necessary computations for calculating the region-based aggregates per supplier. As each multiplication equals the transmission of a share from each of the DCC servers to the others, the communication cost for this part can be calculated by simply multiplying the total number of multiplications (given in Table~\ref{tab:computational complexity}) with the total number of shares exchanged between the DCC servers per multiplication. In our case this is equal to $6 \times |[x]|$, where $|[x]|$ is the size of a share. TRAD and DEP2SA do not have any communication cost in this part.

\begin{table*}[!t]
    \centering  \caption{{Comparison of our Protocol with Existing Aggregation Protocols.}}
  \scalebox{0.99}{
  \begin{tabular}{l c c c c c c c c c c c c c c c}
  \hline
   \textbf{Feature/Property} & 
   \textbf{Our} & \textbf{TRAD} & \textbf{\cite{DEP2SA}} & \textbf{\cite{Danezis2013}} & \textbf{\cite{Rottondi2013IEEE}} & \textbf{\cite{Gope2018}} & \textbf{\cite{Liu2018}} & \textbf{\cite{Lyu2018}} & \textbf{\cite{he2017efficient}} & \textbf{\cite{Li2018}} & \textbf{\cite{Knirsch2015}} & \textbf{\cite{Abdallah2018}} & \textbf{\cite{Shen2017}} & \textbf{\cite{EPPP4SMS}} & \textbf{\cite{TONYALI2018}} \\
  \hline
Real smart metering architecture	& Yes	& Yes 	& Yes 	& No 	& No 	& No 	& No 	& No 	& No 	& No 	& No 	& No 	& No 	& No 	& No \\
Multiple data recipients			& Yes 	& Yes 	& Yes 	& Yes 	& Yes 	& No 	& No 	& No 	& No 	& No 	& No 	& No 	& No 	& No 	& No \\
Flexible subsets of SMs				& Yes 	& Yes 	& Yes 	& Yes	& Yes 	& Yes 	& No 	& No 	& No 	& Yes 	& Yes 	& No 	& Yes 	& Yes 	& No \\
Consumption \& production data 		& Yes 	& Yes 	& No 	& No	& Yes 	& No 	& No 	& No 	& No 	& No 	& No 	& No 	& Yes 	& No 	& No \\
Multiple grid fees calculation 		& Yes 	& Yes 	& Yes 	& Yes	& Yes 	& No 	& No 	& No 	& No 	& No 	& Yes 	& No 	& Yes 	& Yes 	& No \\
Easy supplier switch for users		& Yes 	& Yes 	& Yes 	& No	& No 	& No 	& No 	& No 	& No 	& No 	& No 	& No 	& No 	& No 	& No \\
Easy SM addition/removal			& Yes 	& Yes 	& Yes 	& No	& No 	& No 	& No 	& No 	& No 	& No 	& No 	& Yes 	& No 	& No 	& No \\
\hline
Confidentiality of user data		& Yes 	& No 	& Yes 	& Yes	& Yes 	& No 	& Yes 	& Yes 	& Yes 	& Yes 	& Yes 	& Yes 	& No 	& Yes 	& Yes \\
Internal-attackers resistance		& Yes 	& No 	& Yes 	& No	& Yes 	& No 	& Yes 	& Yes 	& Yes 	& Yes 	& Yes 	& No 	& Yes 	& Yes 	& Yes \\
Privacy of users 					& Yes 	& No 	& Yes 	& Yes	& Yes 	& No 	& Yes 	& Yes 	& Yes 	& Yes 	& Yes 	& Yes 	& Yes 	& Yes 	& Yes \\
Authorisation 						& Yes 	& No 	& Yes 	& No	& Yes 	& Yes 	& Yes 	& No 	& No 	& No 	& No 	& No 	& No 	& No 	& No \\
Fault tolerant 						& Yes 	& No 	& No 	& No	& Yes 	& No 	& Yes 	& Yes 	& No 	& No 	& Yes 	& Yes 	& No 	& No 	& No \\
Collusion-attacks resistance 		& Yes 	& No 	& Yes 	& Yes 	& Yes 	& No 	& Yes 	& No 	& No 	& Yes 	& Yes 	& No 	& Yes 	& No 	& No \\
Require secure channels* 			& Yes 	& Yes 	& No 	& Yes	& Yes	& No 	& No 	& No 	& No 	& No 	& Yes 	& No 	& No 	& Yes 	& No \\
\hline
Computational cost at SMs			& L 	& L 	& H 	& H		& L 	& L 	& H 	& L 	& M 	& H 	& L 	& M 	& H 	& H 	& H \\
Computational cost at recipients	& L 	& L 	& M 	& L 	& L 	& L 	& L 	& L 	& L 	& H 	& L 	& L 	& L 	& H 	& H \\
Communication cost at SMs level		& L 	& L 	& L 	& L		& L 	& L 	& L 	& L 	& L 	& L 	& H 	& L 	& L 	& L 	& H \\
Overall communication cost 			& M 	& L 	& L 	& L 	& M 	& L 	& L 	& L 	& L 	& L 	& H 	& L 	& L 	& L 	& M \\
Scalable 							& Yes 	& Yes 	& Yes 	& Yes 	& Yes 	& Yes 	& Yes 	& Yes 	& Yes 	& Yes 	& No 	& Yes 	& Yes 	& Yes 	& Yes \\
  \hline

  \end{tabular}}
  ~
*This can be provided via the TLS/SSL protocol. ~~~~ L - Low, M - Medium, H - High ~~~~~~~~~~~~~~~~~~~~~~~~~~~~~~~~~~~~~~~~~~~~~~~~~~~~~~~~~~~~
\label{tab:comparison}
\end{table*}


\subsubsection{DCC-to-TSO/DNOs/Suppliers part} In each time slot the DCC servers need to send the computed results to the TSO, DNOs and suppliers. As the output data of NAA, NCAA and NIAA is the same, the communication cost for this part is the same regardless of the aggregation algorithm. In detail, each DCC server has to send (i) $ \textnormal{N}_{\textnormal{d}} \times ([\mathbb{E}_{\textnormal{d}_j,\textnormal{s}_u}^{\textnormal{imp}}] + [\mathbb{E}_{\textnormal{d}_j,\textnormal{s}_u}^{\textnormal{exp}}])$ to each supplier, (ii) $ \textnormal{N}_{\textnormal{s}} \times ([\mathbb{E}_{\textnormal{d}_j,\textnormal{s}_u}^{\textnormal{imp}}] + [\mathbb{E}_{\textnormal{d}_j,\textnormal{s}_u}^{\textnormal{exp}}])$ to each DNO, and $ \textnormal{N}_{\textnormal{d}} \times \textnormal{N}_{\textnormal{s}} \times ([\mathbb{E}_{\textnormal{d}_j,\textnormal{s}_u}^{\textnormal{imp}}] + [\mathbb{E}_{\textnormal{d}_j,\textnormal{s}_u}^{\textnormal{exp}}])$ to the TSO. This results in a total communication cost of $ 9 \times \textnormal{N}_{\textnormal{d}} \times \textnormal{N}_{\textnormal{s}} \times ([\mathbb{E}_{\textnormal{d}_j,\textnormal{s}_u}^{\textnormal{imp}}] + [\mathbb{E}_{\textnormal{d}_j,\textnormal{s}_u}^{\textnormal{exp}}])$. If the suppliers and DNOs trust the TSO (which is usually the case in practice), they could directly obtain the aggregation results from the TSO. In that case, the communication cost will be reduced to $ 3 \times \textnormal{N}_{\textnormal{d}} \times \textnormal{N}_{\textnormal{s}} \times ([\mathbb{E}_{\textnormal{d}_j,\textnormal{s}_u}^{\textnormal{imp}}] + [\mathbb{E}_{\textnormal{d}_j,\textnormal{s}_u}^{\textnormal{exp}}]) + ( \textnormal{N}_{\textnormal{d}} + \textnormal{N}_{\textnormal{s}}) \times {c}_{\textnormal{d}_j,\textnormal{s}_u}$, where ${c}_{\textnormal{d}_j,\textnormal{s}_u}$ is an encrypted message containing the region-supplier based aggregate consumption and production data, i.e., ${c}_{\textnormal{d}_j,\textnormal{s}_u} = Enc_k(\mathbb{E}_{\textnormal{d}_j,\textnormal{s}_u}^{\textnormal{imp}}, \mathbb{E}_{\textnormal{d}_j,\textnormal{s}_u}^{\textnormal{exp}})$. If TRAD (DEP2SA) is used, the DCC sends the respective (ciphertexts of the) aggregate consumption and generation data, $\mathbb{E}_{\textnormal{d}_j,\textnormal{s}_u}^{\textnormal{imp}}, \mathbb{E}_{\textnormal{d}_j,\textnormal{s}_u}^{\textnormal{exp}}$ $(\mathbb{C}_{\textnormal{d}_j,\textnormal{s}_u}^{\textnormal{imp}}, \mathbb{C}_{\textnormal{d}_j,\textnormal{s}_u}^{\textnormal{exp}})$, to the output parties. This results in a communication cost of $ 3 \times \textnormal{N}_{\textnormal{d}} \times \textnormal{N}_{\textnormal{s}} \times (\mathbb{E}_{\textnormal{d}_j,\textnormal{s}_u}^{\textnormal{imp}} + \mathbb{E}_{\textnormal{d}_j,\textnormal{s}_u}^{\textnormal{exp}})$ for TRAD. In the case of DEP2SA. each DNO also sends the respective aggregate of the consumption and generation data, as well as a random number extracted from each ciphertext, to each of the suppliers. This results in a communication cost of  $\textnormal{N}_{\textnormal{d}} \times \textnormal{N}_{\textnormal{s}} \times (2 \times\mathbb{C}_{\textnormal{d}_j,\textnormal{s}_u}^{\textnormal{imp}} + 2 \times\mathbb{C}_{\textnormal{d}_j,\textnormal{s}_u}^{\textnormal{exp}} + \mathbb{E}_{\textnormal{d}_j,\textnormal{s}_u}^{\textnormal{imp}} + \mathbb{E}_{\textnormal{d}_j,\textnormal{s}_u}^{\textnormal{exp}} + \textnormal{r}_{\textnormal{d}_j,\textnormal{s}_u}^{\textnormal{imp}} + \textnormal{r}_{\textnormal{d}_j,\textnormal{s}_u}^{\textnormal{exp}})$.

Table~\ref{tab:communication_overhead} summarises the communication cost of our protocol (with a different aggregation algorithm used), TRAD and DEP2SA, where $|x|$, $|[x]|$, $|C|$, and $|r|$, denote the length of a message, its share, a ciphertext and a random number, respectively. Furthermore, using the parameters from the previous section and setting $|x| = |r| = 32 $, $|[x]| = 63 $, $|{c}| = 128$ and $|{C}| = 1024$, we depict the communication cost of our protocol at each part and the entire smart metering architecture in~Fig.~\ref{fig:communication_overhead} and Fig.~\ref{fig:total_communication_overhead}, respectively. As expected, our protocol has higher communication cost than TRAD due to the privacy protection it offers. In comparison to DEP2SA, only our protocol with NCAA performs better. Regarding the choice of data aggregation algorithms, NCAA is the most efficient. However, this algorithm discloses towards the DCC servers the number of users linked to each supplier. In practice, such disclosure can be tolerated by users. If such disclosures are not accepted, NAA or NIAA should be used. Both algorithms have comparable communication costs, the difference being in the part of the smart metering architecture where the cost is concentrated. In the case of NAA, the main cost incurs at the Between-DCC part, whereas in the case of NIAA -- at the SMs-to-DCC part.


\begin{figure}[t]
\centering
\includegraphics[trim= 0 0 0 0,clip=true,width=2.39in]{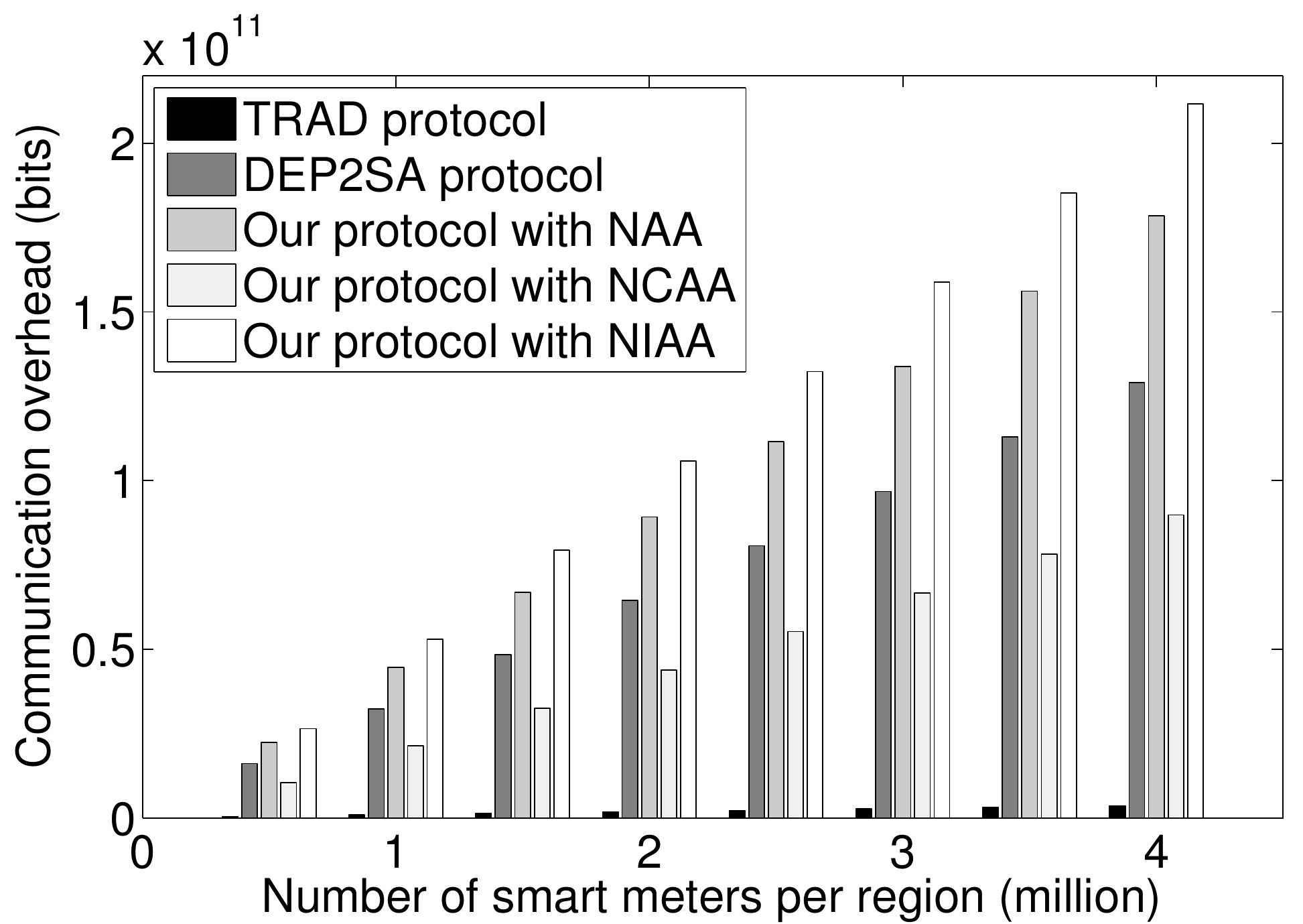}
\caption{The total communication overhead for our protocol.}
\label{fig:total_communication_overhead}
\end{figure}


\subsection{Comparison with Existing Aggregation Protocols}

Table~\ref{tab:comparison} gives an overview of the comparison of our protocol with existing (privacy-preserving) smart metering data aggregation protocols in terms of features, functionalities, security and privacy properties as well as performance. Although most of the existing solutions provide a sufficient protection of users' privacy, they do not support functionalities needed for a real-world smart metering architecture. It is clear that our protocol is the only one that, apart from protecting users' privacy, is readily deployable in a real-world smart metering architecture due to the functionalities it supports as well as its low overhead and good scalability.

\section{Conclusions}
\label{Conclusions}
We introduced an MPC-based protocol for aggregating electricity data, related  to consumption and generation, in a secure and privacy-friendly manner. These data are required for operational purposes such as calculating the transmission, generation and balancing fees. Furthermore, we introduced three data aggregation algorithms that offer different security and performance trade-offs. We also analysed the associated computation and communication costs of all our protocols. Our results indicate the feasibility of our protocol for a setting based on a real smart metering architecture. As a future work, we will aim to have a testbed implementation of our protocol in order to test its performance in a real-world environment.




\section*{Acknowledgment}


This work was supported in part by the Dame Kathleen Ollerenshaw Fellowship awarded by The University of Manchester, by the Research
Council KU Leuven: C16/15/058, by the Flemish
Government through FWO SBO project SNIPPET
S007619, and by H2020-DS-2014-653497 PANORAMIX.



\bibliographystyle{IEEEtran}


\begin{IEEEbiographynophoto}{Mustafa A. Mustafa} received the B.Sc. degree in communications from the Technical University of Varna, Varna, Bulgaria, in 2007, the M.Sc. degree in communications and signal processing from Newcastle University, Newcastle upon Tyne, U.K., in 2010, and the Ph.D. degree in computer science from The University of Manchester, Manchester, U.K., in 2015. 

He is a Dame Kathleen Ollerenshaw Research Fellow in the School of Computer Science at The University of Manchester. Prior that he was a post-doctoral research fellow with the imec-COSIC research group, Department of Electrical Engineering (ESAT), KU Leuven. His research interests include information security, privacy and applied cryptography in areas such as smart grid, e-health, and IoT.
\end{IEEEbiographynophoto}

\begin{IEEEbiographynophoto}{Sara Cleemput} received the M.Sc. degree in biomedical engineering from KU Leuven, Belgium in 2012 and the Ph.D. degree in electrical engineering from KU Leuven, Belgium in 2018. Her research interests include privacy and security for smart electricity grids.
\end{IEEEbiographynophoto}

\begin{IEEEbiographynophoto}{Abdelrahaman Aly} received an engineer degree in computer science and information systems from the Escuela Poli\'ecnica del Ej\'ercito, Quito, Ecuador, in 2010, and the Ph.D. degree in applied mathematics from the Universit\'e catholique de Louvain, Louvain-la-neuve, Belgium, in 2015. 

He is a postdoctoral research fellow with the imec-COSIC research group, Department of Electrical Engineering (ESAT), KU Leuven. His research interests include secure distributed computation.
\end{IEEEbiographynophoto}

\begin{IEEEbiographynophoto}{Aysajan Abidin} received the B.Sc. degree in computational science from Xinjiang University in China in 2006, the M.Sc. degree in engineering mathematics from Chalmers University of Technology in Sweden in 2007, and the Ph.D. degree in information coding from Link\"oping University in Sweden in 2013. He worked on privacy-preserving biometric authentication as a post-doctoral researcher at Chalmers during 2014--2015. 

Since 2015, he is a post-doctoral research fellow with the imec-COSIC research group, Department of Electrical Engineering (ESAT), KU Leuven. His research interests include information security, privacy, authentication, and applied cryptography. 
\end{IEEEbiographynophoto}

\end{document}